%% file: correlated.tex
\documentclass[conference,twocolumn,romanappendices]{IEEEtran}

\IEEEoverridecommandlockouts

\usepackage{silence}
\usepackage{algorithm,algorithmic,amsmath,amssymb,amsthm,bbm,caption,cite,color,comment,graphicx,microtype,setspace,subcaption,tabularx,url}
\usepackage[USenglish]{babel}
\usepackage[T1]{fontenc}
\usepackage[utf8]{inputenc}
\usepackage[hidelinks]{hyperref}
\usepackage{csquotes}
\usepackage{mathtools}
\usepackage{mathrsfs}
\usepackage{mleftright}

\usepackage{enumitem}
\setitemize{itemsep=0mm,topsep=0mm,parsep=0pt,partopsep=0pt}

\usepackage{pgfplots}
\pgfplotsset{compat=newest}

\usepackage{tikz}
\usetikzlibrary{arrows,backgrounds,calc,intersections,decorations.pathreplacing,positioning,shapes,spy}

\makeatletter
\newcommand\subparagraph{%
  \@startsection{subparagraph}{5}
  {\parindent}
  {3.25ex \@plus 1ex \@minus .2ex}
  {-1em}
  {\normalfont\normalsize\bfseries}}
\makeatother
\usepackage{titlesec}
\let\subparagraph\relax

\usepackage{titlesec}
\let\subparagraph\relax
\titlespacing{\section}{0pt}{5pt plus 2pt minus 1pt}{3pt plus 1pt minus 0pt}
\titlespacing{\subsection}{0pt}{4pt plus 2pt minus 1pt}{2pt plus 1pt minus 0pt}
\usepackage[figurename=Fig.,font=small]{caption}

{}
{}
{\newtheorem{lemma}{Lemma}}
{\newtheorem{proposition}{Proposition}}
{\newtheorem{remark}{Remark}}
{}
{}

\input{./Definitions1_corr.tex}

\title{SEP Analysis  of Quantized SIMO Systems with $M$-PSK over Correlated Fading Channels}
\author{
\IEEEauthorblockN{Amila Ravinath, Bikshapathi Gouda, Italo Atzeni, and Antti Tölli}
\IEEEauthorblockA{Centre for Wireless Communications, University of Oulu, Finland \\
E-mail: \{amila.ravinath, bikshapathi.gouda, italo.atzeni, antti.tolli\}@oulu.fi}
\thanks{This work was supported by the Research Council of Finland (336449 Profi6, 348396 HIGH-6G, 357504 EETCAMD, and 369116 6G~Flagship).}}

\begin{document}

\maketitle

\begin{abstract}
The average symbol error probability (SEP) of a phase-quantized single-input multiple-output system with $M$-ary phase-shift keying modulation and maximum ratio combining (MRC) is analyzed under correlated Rayleigh fading and additive white Gaussian noise.  Building on our prior framework for independent and identically distributed Rayleigh fading, we extend the analysis to the spatially correlated case by introducing an asymptotically equivalent MRC combiner  that enables tractable  SEP characterization. Using this approach, we derive closed-form expressions at high signal-to-noise ratio (SNR) that explicitly characterize the diversity and coding gains as functions of the receive correlation structure, phase-quantization resolution, and modulation order, up to a scaling factor bounded between $1$ and $2$. The results show that channel correlation primarily degrades the coding gain, leading to an SNR penalty, while the diversity gain is preserved when the channel covariance matrix is full-rank. The analytical findings are validated through Monte Carlo simulations, demonstrating a tight match across a wide SNR range.
\end{abstract}


\section{Introduction} \label{sec:INTRO}

Low-resolution data conversion is a promising approach for reducing the power
consumption and hardware complexity in multi-antenna receivers. This is
particularly relevant in systems with a large number of radio-frequency chains, where
high-resolution analog-to-digital converters (ADCs) become costly in terms of both power and implementation. As a
result, low-resolution receiver architectures, including 1-bit and few-bit designs,
have attracted significant interest in massive multiple-input multiple-output and related systems
\cite{Jac17,Y_Li_et_al_BLMMSE,Atzeni_2022}. In this context, phase quantization is especially appealing when phase-shift keying (PSK) modulations are used,
since it preserves phase information while enabling simple receiver
structures. Quantized and polar-quantized systems have been studied from both information-theoretic and error-probability perspectives
\cite{Ber22,Wu_Liu_et_al_2023}. However, the symbol error probability (SEP) analysis of
phase-quantized single-input multiple-output (SIMO) systems over correlated fading channels remains a largely open problem.

For the independent and identically distributed (i.i.d.) Rayleigh fading case, \cite{C, J} developed a
tractable framework for the SEP analysis of phase-quantized SIMO systems with
$M$-PSK modulation and maximum ratio combining (MRC). This approach departs from the conventional two-step method, where the conditional SEP is first derived for a given channel realization and then averaged over fading \cite{Gia03,Sim00}. However, the two-step method becomes intractable under coarse quantization since the signal and noise components are no longer separable after quantization. To address this, the approach introduced in \cite{J} jointly exploits the circular symmetry of the noise and fading distributions, enabling direct characterization of the average SEP. Extending this framework to correlated fading is non-trivial, as the symmetry arguments no longer apply directly. This extension is of practical relevance since
spatial correlation is inherent in multi-antenna channels and, in the unquantized case, primarily affects the coding
gain while preserving diversity for full-rank covariance matrices \cite{Gia03,Ko00,Sim00}.

In this paper, we analyze the SEP of an $n$-bit phase-quantized SIMO system
with $M$-PSK modulation over correlated Rayleigh fading channels with additive white Gaussian noise (AWGN), assuming perfect channel state information (CSI) at the receiver. As a relevant special case, a 2-bit phase quantizer corresponds to the conventional 1-bit ADC architecture, where one quantization bit is used for each of the in-phase and quadrature components. Since a direct extension of the MRC-based reformulation in
\cite{J} to correlated channels appears difficult, we introduce an
asymptotic MRC (AMRC) combiner that becomes equivalent to conventional MRC at high
signal-to-noise ratio (SNR) and enables a tractable analysis. Based on this approach, we characterize the
diversity and coding gains for the case $M < 2^n$ and establish the
diversity gain for the boundary case $M = 2^n$. The analytical results are validated via Monte Carlo
simulations, which also illustrate the impact of channel correlation on the SEP.

\emph{Notation.} $(\cdot)^\tran$, $(\cdot)^\herm$, and $(\cdot)^*$ represent the
transpose, Hermitian transpose, and element-wise conjugate operators,
respectively. $a_n$ denotes the $n$th entry of the vector $\mathbf{a}$, whereas $A_{mn}$ is the element in the $m$th row and the
$n$th column of the matrix $\A$. $\0$ represents 
the all-zero vector and $\I_n$ the $n$-dimensional identity matrix. $|\cdot|$ and $\arg(\cdot)$ denote the modulus and argument operators, respectively, with the latter defined in $(-\pi,
\pi]$. $\re{(\cdot)}$ and $\im{(\cdot)}$ extract the real and imaginary parts, respectively, whereas $j=\sqrt{-1}$ is the imaginary unit.
$\real$, $\real_+$, and $\complex$ denote
the sets of real, nonnegative real, and complex numbers, respectively. The Cartesian product of two sets $\setA$ and $\setB$ is denoted by $\setA\times\setB$, whereas the $k$-fold Cartesian
product of a set $\setA$ by $\setA^k$. $n!$ denotes the factorial of $n$. $\Exp{\cdot}$ and $\Prob{\cdot}$
denote the expectation and probability operators, respectively.  $\setC\setN(\0,
\Sigmab)$ represents the zero-mean circularly symmetric complex Gaussian
distribution with covariance matrix $\Sigmab$.  $\setU(a, b)$ denotes the
uniform distribution over the interval $(a, b)$ and $\exp(\lambda)$ the
exponential distribution with mean $\lambda$. $\qfunc{x} \triangleq \frac1{\sqrt
{2\pi}}\int_x^{\infty}e^{-\frac{t^2}2} \, {\mathrm{d}t}$ represents the $Q$-function and
$\Gamma(x) \triangleq \int_0^\infty t^{x-1}e^{-t}\, {\mathrm{d}t}$, $\re\{x\}>0$ the
gamma function. For a random variable $X$, the probability density function
(pdf) is denoted by $f_X(x)$ and the
moment generating function (MGF) by $\setM_X(s)$.  $x \to a^+$
signifies that $x$ tends to $a$ from above.  $\smallO{\cdot}$ and $\bigO{\cdot}$ denote the
little-o and big-o notations, respectively \cite[Pg.~7]{Har58}.


\section{System Model}

In this section, we first present the signal and quantization model, and then describe the signal detection.

\subsection{Signal and Quantization Models}

Consider a SIMO system in which the transmitted  symbol $s$ is drawn uniformly from an $M$-PSK constellation defined as
\begin{align}
\setS_M \triangleq \{e^{j\br{\frac\pi4 + \frac{2\pi }Mi}}: i=0, \ldots, M-1\}.\nonumber
\end{align}
Assuming a receiver equipped with ${N_\textup{r}}$ antennas and phase quantization, the received signal vector prior to quantization is given by
\begin{align}
\y \triangleq [y_1, \ldots, y_{N_\textup{r}}]\Trans = \sqrt{\rho}\h s +\n \in \mathbb{C}^{N_\textup{r}}, \label{eqn: sys_model_original}
\end{align}
where  $\h = [h_1, \ldots, h_{N_\textup{r}}]\Trans \in \mathbb{C}^{N_\textup{r}}$ and $\n = [n_1,  \ldots, n_{N_\textup{r}}]\Trans \sim \mathcal{CN}(\mathbf{0}, \mathbf{I}_{N_\textup{r}})$ represent the channel and AWGN vectors over the $ N_\textup{r} $ receive antennas, respectively. The channel and noise are assumed to be independent. Considering correlated Rayleigh fading, we have $\h \sim \setC\setN(\0, \K)$, where $\K \triangleq \Exp{\h\h^\herm} \in \mathbb{C}^{N_\textup{r} \times N_\textup{r}}$ is the channel covariance matrix that we assume to be full-rank. In this setting, $\rho$ represents the transmit SNR.

The received signal $\y$ in~\eqref{eqn: sys_model_original} is further  $n$-bit phase-quantized. Specifically, let $\mathscr{Q}_{n}:\complex \to \setS_{2^n}$ denote the memoryless $n$-bit phase quantization function defined as \cite{Wu_Liu_et_al_2023}
\begin{align}
\qnbitp{n}{x} \in \argmin_{s\in\setS_{2^n}}\ |s-x|.\label{eqn: Qn_def}
\end{align}
Phase quantization maps each $y_i$ to a $2^n$-PSK constellation point. As an example, for $n=2$, this reduces to the well-known 1-bit quantization scheme \cite{Jac17,Y_Li_et_al_BLMMSE,Atzeni_2022}, in which scalar 1-bit ADCs quantize the real and the imaginary parts of the input separately. By applying $\mathscr{Q}_n$ element-wise, the quantized received signal vector is given by
\begin{align}
\q & \triangleq \qnbitp{n}{\y}\in\setS_{2^n}^{N_\textup{r}}. \label{eqn: Q-nb_first_use}
\end{align}

\subsection{Signal Detection} \label{sec: det}

We consider coherent detection based on $\q$ in \eqref{eqn: Q-nb_first_use} assuming perfect CSI at the receiver; the case with imperfect CSI is left for future work.
A widely used approach is to adopt MRC, as considered for example in \cite{Jac17, Y_Li_et_al_BLMMSE, Atzeni_2022}, followed by minimum-distance detection \cite{Ber22}. The resulting soft-estimated symbol is given by
\begin{align}
\widehat{s}_\mrc\triangleq\qnbitp{m}{\h\Herm\q},\label{eqn: MRC_original}
\end{align}
where $m=\log_2M$ corresponds to the $M$-PSK modulation at the transmitter. For i.i.d. Rayleigh fading, MRC admits a tractable analysis. In particular, following the approach in \cite[Thm.~1]{J}, the SEP analysis can be carried out by exploiting the circular symmetry of  fading and noise, leading to an equivalent representation that is analytically convenient. However, this argument does not directly extend to the case of correlated Rayleigh fading.

To enable a similar analysis in the general correlated case, we introduce the surrogate combiner
\begin{align}
\g \triangleq [g_{1}, \ldots, g_{N_\textup{r}}]\Trans = \C^\half\K^{-\half}\h \in \mathbb{C}^{N_\textup{r}}, \label{eqn: AMRC_combiner}
\end{align}
where $\C \triangleq \Exp{\y\y^\herm} = \rho \K + \I_{N_\textup{r}} \in \mathbb{C}^{N_\textup{r} \times N_\textup{r}}$ is the covariance matrix of the received signal prior to
quantization. Owing to its asymptotic correspondence with MRC at high SNR, we refer to this surrogate combiner as \textit{asymptotic MRC (AMRC)}, which is characterized by the following lemma.

\begin{lemma}\label{lem: asymptotic_MRC}
The AMRC combiner asymptotically coincides with MRC at high SNR, i.e.,
\begin{align}\label{eqn: asymptotic_MRC}
\frac1{\sqrt\rho}\g = \h + \frac1{2\rho}\K^{-1}\h + \smallO{\rho^{-1}} \ \textnormal{as} \ \rho \to \infty.
\end{align}
\end{lemma}

\begin{IEEEproof}
From \eqref{eqn: AMRC_combiner}, we obtain
\begin{align}
\frac1{\sqrt \rho}\g = \br{\I_{N_\textup{r}} + \frac1{\rho}\K^{-1}}^\half\h,
\end{align}
and \eqref{eqn: asymptotic_MRC} follows by applying a Taylor expansion to the term $\br{\I_{N_\textup{r}} + \frac1{\rho}\K^{-1}}^\half$ at
$\rho \to \infty$.
\end{IEEEproof}

\noindent The corresponding soft-estimated symbol is then obtained as
\begin{align}
\widehat{s}_\amrc\triangleq\qnbitp{m}{\g\Herm\q}.\label{eqn: AMRC_original}
\end{align}

\section{SEP Analysis with Correlated Rayleigh Fading}\label{sec: mrc}

Under correlated Rayleigh fading, the SEP is expected to degrade relative to the
i.i.d. case. According to \cite[Prop.~1]{Gia03}, the SEP conditioned on the fading must be
expressible in the form $k\qfunc{\sqrt{\rho V}}$, where $k$ is a constant and
$V$ is a function of the fading, in order to characterize the diversity and
coding gains of the system. For single-input single-output (SISO) systems with
$M$-PSK modulation and AWGN, this requirement can be satisfied by \cite[p.
320, Problem 5.5]{Stuber_mobile_comm}\cite[Thm.~1]{Wu_Liu_et_al_2023} up to some scaling factor. To
extend that result to the SIMO setting, an equivalent SISO representation would
be required, similar to \cite[Thm.~1]{J}. However, for SIMO systems with
correlated Rayleigh fading, such an equivalence does not appear to hold when the
MRC combiner is used. To address this, we establish the required equivalence for
the AMRC combiner which asymptotically coincides with MRC.

Based on \eqref{eqn: AMRC_original}, the corresponding average SEP is defined as
\begin{align}
    \SEPG^{\amrc}&\triangleq \Prob{\widehat{s}_{\amrc}\neq s} \\
                 &= \Prob{\widehat{s}_{\amrc}\neq s|s}, \label{eqn: MRC_SEP_definition}
\end{align}
for any $s\in\setS_M$, where the last equality follows from the symmetry of the $M$-PSK constellation $\setS_M$ and the assumption of equiprobable transmitted symbols. Similar to \cite[Thm.~1]{J}, the following proposition for correlated Rayleigh fading with the AMRC combiner allows the AWGN vector $\n$ to be effectively moved outside the quantization function $\qnbitp{n}{\cdot}$, thereby enabling the subsequent diversity analysis.

\begin{proposition}\label{prop: equiv_new}
For the system model in \eqref{eqn: sys_model_original}, we have
\begin{align}
\nonumber & \Prob{\qnbitp{m}{\g^\herm\qnbitp{n}{\y}} \ne s|s} \\
& = \Prob{\qnbitp{m}{s\br{\qnbitp{n}{\g s}}^\tran{\y^*}} \ne s|s}. \label{eqn: equiv_new}
\end{align}
\end{proposition}

\begin{IEEEproof}
Along the lines of \cite[Thm.~1]{J}, we treat the fading and the noise jointly. In particular, we exploit the fact that the joint vectors $\sqbr{\g^\tran s, \y^\tran}^\tran \in \mathbb{C}^{2 N_\textup{r}}$ and $\sqbr{\y^\tran, \g^\tran s}^\tran \in \mathbb{C}^{2 N_\textup{r}}$ follow the same distribution
\begin{align}
    \mathcal{CN}\br{\0, \begin{pmatrix}\C&
\sqrt\rho\C^\half\K^\half\\\sqrt\rho\C^\half\K^\half & \C\end{pmatrix}}
\end{align}
when conditioned on $s$. This follows from
\begin{align}
  \Exp{\g s|s} &= \C^\half\K^{-\half}\Exp{\h|s}s
  \\          &= \C^\half\K^{-\half}\Exp{\h}s = \0,
  \\ \Exp{\y|s}  &= \sqrt{\rho}\Exp{\h|s}s + \Exp{\n|s}
  \\          &= \sqrt{\rho}\Exp{\h}s + \Exp{\n} = \0,
  \\ \Exp{\g s\br{\g s}^\herm|s} &= \Exp{\y\y^\herm|s} = \C,
  \\\Exp{\g s \y^\herm |s} &= \Exp{\y\br{\g s}^\herm|s} = \sqrt{\rho}\C^\half\K^\half,
\end{align}
where we used the mutual independence of the channel vector $\h$, the AWGN vector $\n$, and the transmitted symbol $s$. Now, observe that the left-hand side of \eqref{eqn: equiv_new} can be written as $\Prob{\qnbitp{m}{s(\g s)^\herm\qnbitp{n}{\y}} \ne s|s}$. Since $\sqbr{\g^\tran s, \y^\tran}^\tran$ and $\sqbr{\y^\tran, \g^\tran s}^\tran$ are identically distributed, we can interchange their roles. The result then follows from the identity
\begin{align}
\y^\herm \qnbitp{n}{\g s} = (\qnbitp{n}{\g s})^\tran\y^*.
\end{align}
\end{IEEEproof}

\begin{remark}
    Unlike the i.i.d. counterpart of Proposition~\ref{prop: equiv_new} given in \cite[Thm.~1]{J}, the equivalent representation in the right-hand side of \eqref{eqn: equiv_new} is symbol-dependent, as the quantized received signal includes the transmitted symbol $s$. This makes the representation impractical from an implementation perspective. Its role here is purely analytical, as it enables a reformulation that facilitates the subsequent diversity analysis.
\end{remark}

%
%

Proposition~\ref{prop: equiv_new} and Lemma~\ref{lem: asymptotic_MRC} extend \cite[Prop.~1]{J} to the case of correlated Rayleigh fading at
high SNR, leading to the following result.

\begin{proposition}\label{prop: corr_Wu_bound}
Let $\SEPG^\amrc_{\h}$ denote the SEP conditioned on the channel realization and define $U\triangleq\frac2{N_\textup{r}}\br{\sum_{i}Z_i}^2$, with $Z_i \triangleq |h_i|\sin\br{\frac\pi M - \widetilde\theta_i}$ and $\widetilde\theta_i \triangleq \arg\br{\qnbitp{n}{h_i^*}h_i}$. We have
\begin{align}\label{eqn: wu_bound}
\SEPG^\amrc_{\h} = k\qfunc{\sqrt{\rho U}} + \smallO{\qfunc{\sqrt{\rho U}}}  \ \textnormal{as} \ \rho \to \infty,
\end{align}
for some $k\in[1, 2]$, i.e., the asymptotic SEP is determined up to a scaling factor bounded between $1$ and $2$.
\end{proposition}

\begin{IEEEproof}
Let us fix $s = e^{j\frac\pi4}$ without loss of generality. From the right-hand side of \eqref{eqn: equiv_new}, we extract an equivalent SISO system expressed as
\begin{align}
\widetilde y \triangleq 
\sqrt{\rho} \widetilde h e^{j\frac\pi4} + \widetilde n,\label{eqn: SISO}
\end{align}
with
\begin{align}
\widetilde h & \triangleq \sum_i e^{-j\frac\pi4}\qnbitp{n}{g_i e^{j\frac\pi4}}h_i^*, \\
\widetilde n & \triangleq \sum_i e^{j\frac\pi4}\qnbitp{n}{g_i e^{j\frac\pi4}} n_i^*.
\end{align}
Conditioned on $\h$, we have $\widetilde n \sim \mathcal{CN}(0, N_\textup{r})$,
since $\n$ is circularly symmetric and each term
$e^{j\frac\pi4}\qnbitp{n}{f_ie^{j\frac\pi4}}$ has unit magnitude corresponding
to a phase rotation of $n_i^*$. Now, \eqref{eqn: SISO} yields the bound \cite[p.
320, Problem 5.5]{Stuber_mobile_comm}\cite[Thm.~1]{Wu_Liu_et_al_2023}
\begin{align}
\qfunc{\sqrt{\frac{2\rho}{N_\textup{r}}}\eta\sin\frac\pi M} \le \SEPG^\amrc_{\h} \le 2\qfunc{\sqrt{\frac{2\rho}{N_\textup{r}}}\eta\sin\frac\pi M},
\end{align}
with 
\begin{align}
\eta \triangleq \re{(\widetilde h)} - \abs{\im{(\widetilde h)}}\cot\frac\pi M.
\end{align}
Based on Lemma~\ref{lem: asymptotic_MRC}, we have
\begin{align}
    \qnbitp{n}{g_i e^{j\frac\pi4}} = \qnbitp{n}{h_i e^{j\frac\pi4} + \bigO{\rho^{-1}}} = \qnbitp{n}{h_i e^{j\frac\pi4}},\label{eqn: g_high_SNR}
    \end{align}
    with probability $1$ for sufficiently large $\rho$. Finally, substituting \eqref{eqn: g_high_SNR} into $\widetilde h$ and $\eta$ yields \eqref{eqn: wu_bound}.
\end{IEEEproof}

\begin{remark}
The cases $M<2^n$ and $M=2^n$ exhibit fundamentally different behaviors due to the angular terms $\sin\br{\frac\pi M - \widetilde\theta_i}$. For $M<2^n$, these terms remain strictly positive for all the realizations, so all the antennas contribute equally to the decision statistics. 
In contrast, for $M=2^n$, the angular terms can vanish at the boundaries of the phase-quantization regions, leading to channel realizations for which the contributions of individual antennas become arbitrarily small. As a result, the dominant error events arise from different mechanisms in the two cases, motivating the use of different analytical approaches.
\end{remark}

Next, we present our main results on the diversity and coding gains for $M<2^n$ and $M=2^n$ in Sections~\ref{sec:M<2^n} and~\ref{sec:M=2^n}, respectively.

\subsection{Case \texorpdfstring{$M<2^n$}{Inequality Case}} \label{sec:M<2^n}

The diversity and coding gains for $M<2^n$ are characterized in the following proposition.

\begin{proposition}\label{prop: high_SNR_Mle2ton}
The diversity gain and the coding gain up to a scaling factor $k\in[1, 2]$ of a phase-quantized SIMO system with $M<2^n$ and correlated Rayleigh fading are given by
\begin{align}
G_\textup{d} = N_\textup{r}\label{eqn: diversity_gain_Mle2ton}
\end{align}
and
\begin{align}
G_\textup{c} &= \br{\frac{2^{n{N_\textup{r}}-1}kN_\textup{r}^{N_\textup{r}}}{{\pi}^{N_\textup{r} + \frac12}(2{N_\textup{r}})!\det(\K)}\br{\cot\br{\frac\pi M - \frac\pi{2^n}} \mright.\mright.\nonumber \\&\phantom= \ \mleft.\mleft.- \cot\br{\frac\pi M + \frac\pi{2^n}}}^{N_\textup{r}}\gammaf{{N_\textup{r}} + \frac12}}^{-\frac1{N_\textup{r}}}, \label{eqn: coding_gain_Mle2ton}
\end{align}
respectively.
\end{proposition}

\begin{IEEEproof}
The MGF of $T \triangleq \sum_i Z_i$ is given by
\begin{align}
\setM_T(s) &= \Exp{e^{-sT}}
\\         &= \int_{\h \in \complex^{N_\textup{r}}}e^{-sT}\frac{1}{\pi^{N_\textup{r}}\det(\K)}e^{-\h^\herm\K^{-1}\h}\, \mathrm{d}\h. \label{eq:MGF1}
\end{align}
Define $\setA \triangleq \real_{+}^{N_\textup{r}}\times (-\pi,
\pi)^{N_\textup{r}}$, $\bar{\K} \triangleq \K^{-1}$, and $a(\theta) \triangleq
\sin\br{\frac\pi M - \theta}$. Applying the change of variables $h_i =
r_ie^{j\theta_i}$ allows to rewrite \eqref{eq:MGF1} as
\begin{align}
\setM_T(s) &= \frac{1}{\pi^{N_\textup{r}}\det(\K)}\idotsint_{(\r, \thetab)\in\setA}e^{-s\sum_{i}r_ia(\widetilde\theta_i)}\nonumber\\&\phantom= \ \times e^{- \sum_{i, j}\bar K_{ij}r_ir_je^{j(\theta_i - \theta_j)}} \prod_{i}r_i \, \mathrm{d}r_i \, \mathrm{d}\theta_i.
\end{align}
Further applying the change of variables $s r_i = s_i\ge 0$ yields
\begin{align}
\setM_T(s) &= \frac{s^{-2N_\textup{r}}}{\pi^{N_\textup{r}}\det(\K)}\idotsint_{(\s, \thetab)\in\setA}e^{-\sum_{i}s_ia(\widetilde\theta_i)}\nonumber\\&\phantom= \ \times e^{ - \frac{1}{s^2}\sum_{i, j}\bar K_{ij}s_is_je^{j(\theta_i - \theta_j)}} \prod_{i}s_i \, \mathrm{d}s_i \, \mathrm{d}\theta_i. \label{eq:MGF2}
\end{align}
Since our interest lies in the limit of $s\to \infty$, we have 
\begin{align}
e^{-\frac{1}{s^2}\sum_{i, j}\bar K_{ij}s_is_je^{j(\theta_i - \theta_j)}} = 1 + \smallO{1} \ \textnormal{as} \ s\to\infty,
\end{align}
which allows to simplify \eqref{eq:MGF2} as
\begin{align}
\setM_T(s) 
&= \frac{s^{-2N_\textup{r}}}{\pi^{N_\textup{r}}\det(\K)}\prod_{i}\int_{-\pi}^{\pi}\int_{0}^{\infty} s_ie^{-s_ia(\widetilde\theta_i)} \, \mathrm{d}s_i \, \mathrm{d}\theta_i \nonumber\\&\phantom= \ + \smallO{s^{-2N_\textup{r}}} \ \textnormal{as} \ s\to\infty.\label{eqn: coding_proof}
\end{align}
Since for $M < 2^n$ we have $a(\widetilde\theta_i)>0$, $\forall i = 1,
\ldots, N_\textup{r}$, the inner integrals of the leading term of \eqref{eqn: coding_proof} evaluate to $(a(\widetilde\theta_i))^{-2}$, $\forall i=1, \ldots, N_\textup{r}$. Now, apply the change of variables $\widetilde\theta_i = \arg\br{\qnbitp{n}{h_i^*}h_i} = \arg\br{\qnbitp{n}{e^{-j\theta_i}}e^{j\theta_i}}$, binning $\theta_i\in(-\pi, \pi)$ to $2^n$ bins, with $\mathrm{d}\theta_i = \mathrm{d}\widetilde\theta_i$, $\forall i=1, \ldots, N_\textup{r}$. We obtain
\begin{align}
\setM_T(s) 
&= \frac{s^{-2N_\textup{r}}}{\pi^{N_\textup{r}}\det(\K)}\prod_{i}\int_{-\frac\pi{2^n}}^{\frac\pi{2^n}}2^n\cosec^2\br{\frac\pi M - \widetilde\theta_i} \, \mathrm{d}\widetilde\theta_i \nonumber\\&\phantom= \ + \smallO{s^{-2N_\textup{r}}}
\\         &= b s^{-2{N_\textup{r}}} + \smallO{s^{-2{N_\textup{r}}}} \ \textnormal{as} \ s\to\infty,
\end{align}
with $b \triangleq \frac{2^{nN_\textup{r}}}{\pi^{N_\textup{r}}\det(\K)}\br{\cot\br{\frac \pi M - \frac \pi{2^n}}- \cot\br{\frac \pi M + \frac \pi{2^n}}}^{N_\textup{r}}$.
Thus, the pdf of $T$ is given by \cite[Thm.~35.1]{Doe63}
\begin{align}
f_T(x) = \frac{b}{(2N_\textup{r}-1)!}x^{2N_\textup{r}-1} + \smallO{x^{2N_\textup{r}-1}} \ \textnormal{as} \ x\to0^+.
\end{align}
Consequently, the pdf of $U$ is given by
\begin{align}
f_U(x) 
&= \frac{1}{2}\br{\frac{N_\textup{r}}{2}}^{N_\textup{r}}\frac{b}{(2N_\textup{r}-1)!}x^{N_\textup{r}-1}  + \smallO{x^{N_\textup{r}-1}}.
\end{align}
Finally, invoke \cite[Prop.~1]{Gia03} to obtain \eqref{eqn:  diversity_gain_Mle2ton}--\eqref{eqn: coding_gain_Mle2ton}.
\end{IEEEproof}

\begin{remark}
From \eqref{eqn: coding_gain_Mle2ton}, we observe that, for $M < 2^n$, the impact of channel correlation on the coding gain is quantified explicitly by the
factor $\det(\K)$. In particular, compared with the i.i.d. case in \cite[Eq.~(24)]{J}, which implicitly assumes $\det(\K)=1$, channel correlation induces an SNR penalty that scales inversely with $\det(\K)$. Hence, while the diversity gain remains unchanged for full-rank $\K$, the coding gain is reduced due to the loss of spatial independence.
\end{remark}

Fig.~\ref{fig: wu} verifies Propositions~\ref{prop: corr_Wu_bound} and~\ref{prop: high_SNR_Mle2ton}, with $N_\textup{r} = 4$, $M=8$, and $n=4$. The $(i,j)$th element of the channel covariance matrix is defined as $K_{ij}=|\alpha|^{|i-j|}e^{j\phi(i-j)}$, with $\phi = \frac{\pi}{4}$ and $\alpha \in \{0, 0.7, 0.9\}$. $\alpha=0$ (or, more precisely, $\alpha\to0^+$) corresponds to the i.i.d. case with $\K = \I_{N_\textup{r}}$. The lower and upper bounds from \eqref{eqn: wu_bound} are obtained by fixing $k=1$ and $k=2$, respectively, and are also computed via Monte Carlo simulations. As expected, AMRC is lower-bounded by MRC at low SNR and coincides with MRC at high SNR. At high SNR, the upper bound from \eqref{eqn: wu_bound} is tight with respect to AMRC, while the upper bound determined via \eqref{eqn: diversity_gain_Mle2ton}--\eqref{eqn: coding_gain_Mle2ton} with $k=2$ matches the asymptotic behavior of the simulated  SEP. When $\alpha$ increases, the SEP degrades as quantified by $\det(\K)$ at high SNR.

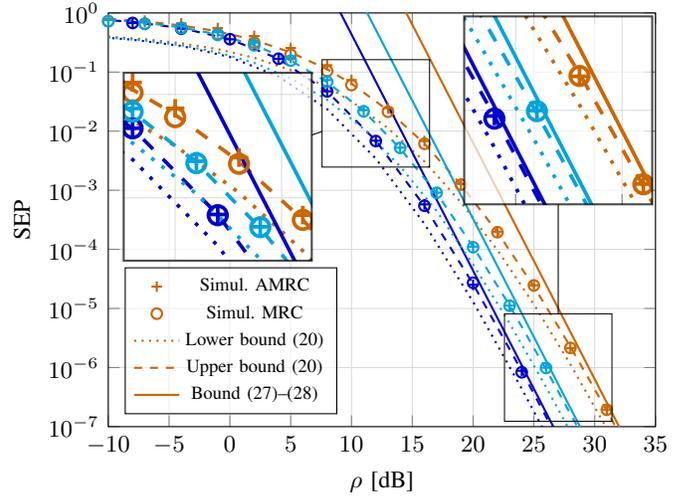
\begin{figure}
\centering
\input {figs/wu.tex}
\caption{SEP versus $\rho$ for $M<2^n$, with $N_\textup{r} = 4$, $M=8$, and $n=4$. The channel covariance matrix is defined as $K_{ij}=|\alpha|^{|i-j|}e^{j\phi(i-j)}$,
with $\phi = \frac{\pi}{4}$ and $\alpha = 0$ {\color{blue!80!black}$\blacksquare$}, $\alpha = 0.7$ {\color{cyan!90!black}$\blacksquare$}, or $\alpha = 0.9$ {\color{orange!80!black}$\blacksquare$}.}
\label{fig: wu}
\end{figure}

\subsection{Case \texorpdfstring{$M=2^n$}{Equality Case}} \label{sec:M=2^n}

In the following, we first derive the diversity gain for $M=2^n$ and then discuss the corresponding coding gain.

\begin{proposition}\label{prop: high_SNR_equal}
The diversity gain of a phase-quantized SIMO system with $M=2^n$ and correlated Rayleigh fading is given~by
\begin{align}
G_\textup{d} = \frac{N_\textup{r}}{2}.\label{eqn: diversity_gain_equal}
\end{align}
\end{proposition}
\begin{IEEEproof}
The pdf of $T$ is given by
\begin{align}
f_T(x) = \frac {\mathrm{d}} {\mathrm{d}x}\int_{T\le x}\frac1{\pi^{N_\textup{r}}\det(\K)}e^{-\h^\herm\K^{-1}\h}\, \mathrm{d}\h.\label{eqn: pdf_T}
\end{align}
Since we are interested in the event $\{\h: T\le x\}$ in the limit of $x\to0^+$, only
the terms $Z_i$ for which $\sin\br{\frac\pi M - \widetilde\theta_i}$ is close to $0$ contribute significantly. Hence, we focus on the local behavior around the critical angles. Define $u_i\triangleq \frac\pi M - \widetilde\theta_i$.
Using a local approximation for a single antenna \cite{Ble86}, we obtain
\begin{align}
& \Prob{r_i\sin u_i \le x} \nonumber \\
&= \Prob{r_i u_i \le x} + \smallO{\Prob{r_i u_i \le x}}
\\                       &=       \br{A_{\K, M, n} + \smallO{1}}\iint_{(r_i,u_i) \in \setB} r_i \, \mathrm{d}r_i \, \mathrm{d}u_i
\\                       &=       A_{\K, M, n}Cx + \smallO{x} \ \textnormal{as} \ x\to0^+,
\end{align}
with $\setB \triangleq \{(r, u)\in\real_+^2 : r<C,ru<x\}$, where  $A_{\K, M, n}$ and $C$ are positive constants.
Applying the same argument across all the $N_\textup{r}$ antennas leads to
\begin{align}
\Prob{T \le x} = B_{\K, N_\textup{r}, M, n}x^{N_\textup{r}} + \smallO{x^{N_\textup{r}}} \ \textnormal{as} \ x\to0^+,
\end{align}
where $B_{\K, N_\textup{r}, M, n}$ is a positive constant due to
\begin{align}
0 < e^{-\kappa N_\textup{r}^2C^2} \le e^{-\h^\herm\K^{-1}\h} \le 1,
\end{align}
$\forall i=1, \ldots, N_\textup{r}$, $\forall (r_i, u_i) \in \setB$, with $\kappa \triangleq \max_{i, j} \abs{\bar K_{ij}}$.  Therefore, the pdf of $U$ is given by
\begin{align}
\hspace{-1mm} f_U(x) 
&= \frac{1}{2}D_{\K, N_\textup{r}, M, n}x^{\frac{N_\textup{r}}{2}-1} + \smallO{x^{\frac{N_\textup{r}}{2}-1}}, \ \textnormal{as} \ x\to0^+,
\end{align}
where $D_{\K, N_\textup{r}, M, n}$ is a positive constant.
Finally, applying \cite[Prop.~1]{Gia03} to the asymptotic behavior of $f_U(x)$ yields \eqref{eqn: diversity_gain_equal}.
\end{IEEEproof}

\begin{remark}
The coding gain for $M=2^n$ cannot be readily obtained using the above approach. The main difficulty stems from the integral in \eqref{eqn: pdf_T}, which does not admit a tractable evaluation in the limit of $x\to 0^+$, and thus the constant $D_{\K, N_\textup{r}, M, n}$ cannot be expressed in closed form. This suggests that a different analytical approach may be required to rigorously characterize the coding gain.
\end{remark}

Inspired by \eqref{eqn: coding_gain_Mle2ton}, we can obtain a heuristic estimate for the coding gain as (cf. \cite[Eq.~(24a)]{J})
    \begin{align}
        \hat{G}_\textup{c} = \br{\!k\frac{{N_\textup{r}}^{\frac{N_\textup{r}}{2}}}{{N_\textup{r}}!\sqrt{\det(\K)}}2^{-{N_\textup{r}}-1}\pi^{-\frac{{N_\textup{r}}+1}{2}}M^{N_\textup{r}}\gammaf{\frac{{N_\textup{r}}+1}{2}}\!}^{-\frac2{N_\textup{r}}},\label{eqn: coding_gain_equal}
    \end{align}
   with $k\in[1, 2]$. 
   The expression in \eqref{eqn: coding_gain_equal} is obtained heuristically by
   using the diversity gain in \eqref{eqn: diversity_gain_equal}, matching it to
   the generic asymptotic SEP form in \cite[Prop.~1]{Gia03}, and then
   substituting the scaling factor suggested by the i.i.d. result together with
   the correlation-dependent term $\sqrt{\det(\K)}$. For $\K =
   \I_{N_\textup{r}}$, \eqref{eqn: coding_gain_equal} corresponds to the exact
   coding gain of the i.i.d. counterpart. It should be emphasized that the expression
   $(\hat G_{\rm c}\rho)^{-G_{\rm d}}$ based on \eqref{eqn: coding_gain_equal} and
   \eqref{eqn: diversity_gain_equal} does not yield a rigorous upper or lower
   bound on the SEP, but rather serves as a heuristic approximation.

Fig.~\ref{fig: wu2} corroborates Proposition~\ref{prop: high_SNR_equal}, with $N_\textup{r} = 4$, $M\in\{4, 8\}$, and the same correlation model used in Fig.~\ref{fig: wu}, with $\phi = \frac{\pi}{4}$ and $\alpha = 0.8$. Although the heuristic bound based on \eqref{eqn: diversity_gain_equal} and \eqref{eqn: coding_gain_equal} follows the simulated SEP closely, it is not as tight as for the case $M<2^n$. However, the simulated SEP corresponds to the diversity gain quantified by \eqref{eqn: diversity_gain_equal}.

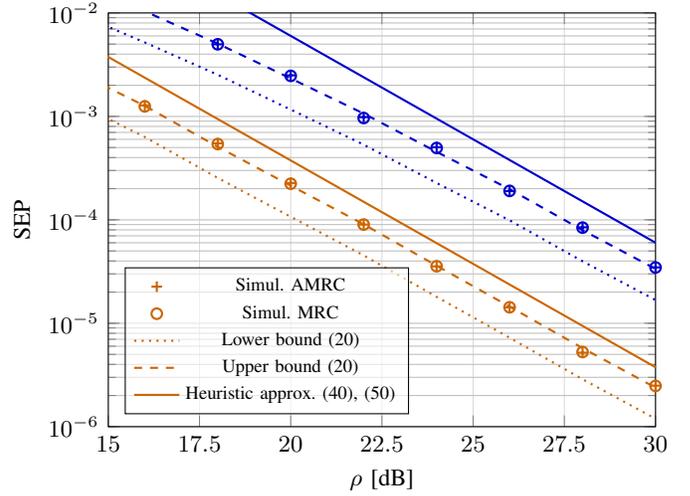
\begin{figure}
\centering
\input {figs/wu2.tex}
\caption{SEP versus $\rho$ for $M=2^n$, with $N_\textup{r} = 4$, and $M=4$ {\color{orange!80!black}$\blacksquare$} or $M=8$ {\color{blue!80!black}$\blacksquare$}. The channel covariance matrix is defined as $K_{ij}=|\alpha|^{|i-j|}e^{j\phi(i-j)}$,
with $\phi = \frac{\pi}{4}$ and $\alpha = 0.8$.}
\label{fig: wu2}
\end{figure}

\section{Conclusions} \label{sec:CONCL}

We analyzed the SEP of an $n$-bit phase-quantized SIMO system with
$M$-PSK modulation over correlated Rayleigh fading channels. To this end, we introduced the AMRC combiner, which becomes equivalent
to MRC at high SNR while enabling a tractable asymptotic analysis. For full-rank channel covariance matrices, we showed that, for the case $M < 2^n$, the system achieves
full diversity, namely $G_\textup{d} = N_\textup{r}$, and we derived the
corresponding coding gain. For the boundary case $M = 2^n$, we established that
the diversity gain reduces to $G_\textup{d} = \frac{N_\textup{r}}{2}$. Simulation results confirmed that the AMRC-based characterization closely
matches the performance of conventional MRC in the high-SNR regime.
Overall, these results extend SEP analysis of phase-quantized SIMO systems from the
i.i.d. setting to correlated fading channels. Future work may consider
characterizing the coding gain for $M = 2^n$ as well as extensions to imperfect CSI and multiple users.

\appendices

\bibliographystyle{IEEEtran}
\bibliography{IEEEabbr_corr, corr_refs}
\end{document}

%% file: Definitions1_corr.tex
\def\real{\ensuremath{\mathbb{R}}}
\def\complex{\ensuremath{\mathbb{C}}}

\def\setS{\ensuremath{\mathcal{S}}}
\def\setC{\ensuremath{\mathcal{C}}}

\def\setA{\ensuremath{\mathcal{A}}}
\def\setB{\ensuremath{\mathcal{B}}}
\def\setM{\ensuremath{\mathcal{M}}}

\def\setN{\ensuremath{\mathcal{N}}}
\def\setU{\ensuremath{\mathcal{U}}}

\def\A{\ensuremath{\mathbf{A}}}
\def\C{\ensuremath{\mathbf{C}}}

\def\n{\ensuremath{\mathbf{n}}}

\def\r{\ensuremath{\mathbf{r}}}
\def\s{\ensuremath{\mathbf{s}}}

\def\y{\ensuremath{\mathbf{y}}}

\def\g{\ensuremath{\mathbf{g}}}
\def\h{\ensuremath{\mathbf{h}}}

\def\I{\ensuremath{\mathbf{I}}}
\def\K{\ensuremath{\mathbf{K}}}

\def\q{\ensuremath{\mathbf{q}}}
\def\1{\ensuremath{\mathbf{1}}}
\def\0{\ensuremath{\mathbf{0}}}

\def\herm{\ensuremath{{\rm H}}}
\def\tran{\ensuremath{{\rm T}}}

\newcommand\abs[1]{\lvert#1\rvert}

\def\max{\ensuremath{\mathrm{max}}}

\newcommand{\re}{{\mathfrak{R}}}
\newcommand{\im}{{\mathfrak{I}}}

\newcommand\Exp[1]{\mathbb{E}\mleft\{#1\mright\}}

\newcommand\Prob[1]{\mathbb{P}\mleft\{#1\mright\}}

\newcommand\qfunc[1]{Q\mleft(#1\mright)}

\newcommand\gammaf[1]{\Gamma\mleft(#1\mright)}

\newcommand\qnbitp[2]{\ensuremath{\mathscr{Q}_{#1}\mleft(#2\mright)}}

\def\SEPG{\mathscr{P}}

\def\mrc{{\scriptscriptstyle \textup {MRC}}}
\def\amrc{{\scriptscriptstyle \textup {AMRC}}}

\newcommand\sqbr[1]{\mleft[#1\mright]}
\newcommand\br[1]{\mleft(#1\mright)}

\newcommand{\Herm}{^{\mathrm{H}}}
\newcommand{\Trans}{^\mathrm{T}}
\newcommand{\cosec}{\mathrm{cosec}}

\def\thetab{\ensuremath{\boldsymbol{\theta}}}
\def\Sigmab{\ensuremath{\boldsymbol{\Sigma}}}

\DeclareMathOperator*{\argmin}{argmin}

\newcommand{\bigO}[1]{\mathcal{O}\mleft(#1\mright)}
\newcommand\smallO[1]{
    \mathchoice
    {
      {\scriptstyle\mathcal{O}}{\mleft(#1\mright)}
    }
    {
      {\scriptstyle\mathcal{O}}{\mleft(#1\mright)}
    }
    {
      {\scriptscriptstyle\mathcal{O}}{\mleft(#1\mright)}
    }
    {
      \scalebox{0.8}{$\scriptscriptstyle\mathcal{O}$}{\mleft(#1\mright)}
    }
}

\def\half{\frac12}

%% file: figs/wu.tex
\begin{tikzpicture} [auto, spy using outlines=
	{rectangle, magnification=1.5, connect spies}]
\begin{semilogyaxis}
[
width=\columnwidth,
height=0.8\columnwidth,
ymax=1,
ymin=1e-7,
xmin=-10,
xmax=35,
xtick={-10,-5, ...,35},
ytick={10^0,10^-1,10^-2,10^-3,10^-4,10^-5,10^-6,10^-7,10^-8},
xlabel = {$\rho$~[dB]},
ylabel = {SEP},
ylabel near ticks,
x label style={font=\small},
y label style={font=\small},
ticklabel style={font=\small},
legend style = {font=\scriptsize, fill=white, fill opacity=0.6, text opacity=1},
legend pos = south west,
grid=both,
major grid style={line width=.2pt,draw=gray!30},
every axis plot/.append style={thick},
mark options = {solid},
mark size = 2pt,
]

\addplot    [color=orange!80!black, only marks, mark=+] table [ y=amrc, x=rho, col sep=comma ] {data/simoNr4M8n4WuCorr09.txt}; \addlegendentry{Simul. AMRC}
\addplot    [color=orange!80!black, only marks, mark=o] table [ y=mrc, x=rho, col sep=comma ]  {data/simoNr4M8n4WuCorr09.txt}; \addlegendentry{Simul. MRC}
\addplot    [color=orange!80!black, dotted] table [ y=wu_lb, x=rho, col sep=comma ]            {data/simoNr4M8n4WuCorr09.txt}; \addlegendentry{Lower bound \eqref{eqn: wu_bound}}
\addplot    [color=orange!80!black, dashed] table [ y=wu_ub, x=rho, col sep=comma ]            {data/simoNr4M8n4WuCorr09.txt}; \addlegendentry{Upper bound \eqref{eqn: wu_bound}}
\addplot    [color=orange!80!black, solid] table [ y=bound, x=rho, col sep=comma ]             {data/simoNr4M8n4WuCorr09.txt}; \addlegendentry{Upper bound \eqref{eqn: diversity_gain_Mle2ton}--\eqref{eqn: coding_gain_Mle2ton}}

\addplot    [color=blue!80!black, only marks, mark=+] table [ y=amrc, x=rho, col sep=comma ] {data/simoNr4M8n4WuCorr0.txt};
\addplot    [color=blue!80!black, only marks, mark=o] table [ y=mrc, x=rho, col sep=comma ]  {data/simoNr4M8n4WuCorr0.txt};
\addplot    [color=blue!80!black, dotted] table [ y=wu_lb, x=rho, col sep=comma ]            {data/simoNr4M8n4WuCorr0.txt};
\addplot    [color=blue!80!black, dashed] table [ y=wu_ub, x=rho, col sep=comma ]            {data/simoNr4M8n4WuCorr0.txt};
\addplot    [color=blue!80!black, solid] table [ y=bound, x=rho, col sep=comma ]             {data/simoNr4M8n4WuCorr0.txt};

\addplot    [color=cyan!90!black, only marks, mark=+] table [ y=amrc, x=rho, col sep=comma ] {data/simoNr4M8n4WuCorr07.txt};
\addplot    [color=cyan!90!black, only marks, mark=o] table [ y=mrc, x=rho, col sep=comma ]  {data/simoNr4M8n4WuCorr07.txt};
\addplot    [color=cyan!90!black, dotted] table [ y=wu_lb, x=rho, col sep=comma ]            {data/simoNr4M8n4WuCorr07.txt};
\addplot    [color=cyan!90!black, dashed] table [ y=wu_ub, x=rho, col sep=comma ]            {data/simoNr4M8n4WuCorr07.txt};
\addplot    [color=cyan!90!black, solid] table [ y=bound, x=rho, col sep=comma ]             {data/simoNr4M8n4WuCorr07.txt};

\coordinate (spypoint) at (axis cs:27,1e-6);
  \coordinate (magnifyglass) at (axis cs:27,2.25e-2);

  \coordinate (spypoint2) at (axis cs:12,2e-2);
  \coordinate (magnifyglass2) at (axis cs:-1,2.5e-3);
\end{semilogyaxis}
\spy [magnification=1.75, size=2.5cm] on (spypoint)
   in node[fill=white, opacity=.7] at (magnifyglass);

   \spy [magnification=1.75, size=2.5cm] on (spypoint2)
   in node[fill=white, opacity=.7] at (magnifyglass2);
\end{tikzpicture}

%% file: figs/wu2.tex
\begin{tikzpicture} [auto]
\begin{semilogyaxis}
[
width=\columnwidth,
height=0.8\columnwidth,
ymax=1e-2,
ymin=1e-6,
xmin=15,
xmax=30,
xtick={15,17.5,...,30},
xlabel = {$\rho$~[dB]},
ylabel = {SEP},
ylabel near ticks,
x label style={font=\small},
y label style={font=\small},
ticklabel style={font=\small},
legend style = {font=\scriptsize, fill=white, fill opacity=0.6, text opacity=1},
legend pos = south west,
grid=both,
major grid style={line width=.2pt,draw=gray!30},
every axis plot/.append style={thick},
mark options = {solid},
mark size = 2pt,
]

\addplot    [color=orange!80!black, only marks, mark=+] table [ y=amrc, x=rho, col sep=comma ] {data/simoNr4M4n2WuCorr08.txt}; \addlegendentry{Simul. AMRC}
\addplot    [color=orange!80!black, only marks, mark=o] table [ y=mrc, x=rho, col sep=comma ]  {data/simoNr4M4n2WuCorr08.txt}; \addlegendentry{Simul. MRC}
\addplot    [color=orange!80!black, dotted] table [ y=wu_lb, x=rho, col sep=comma ]            {data/simoNr4M4n2WuCorr08.txt}; \addlegendentry{Lower bound \eqref{eqn: wu_bound}}
\addplot    [color=orange!80!black, dashed] table [ y=wu_ub, x=rho, col sep=comma ]            {data/simoNr4M4n2WuCorr08.txt}; \addlegendentry{Upper bound \eqref{eqn: wu_bound}}
\addplot    [color=orange!80!black, solid] table [ y=bound, x=rho, col sep=comma ]             {data/simoNr4M4n2WuCorr08.txt}; \addlegendentry{Heuristic approx. \eqref{eqn: diversity_gain_equal}, \eqref{eqn: coding_gain_equal}}

\addplot    [color=blue!80!black, only marks, mark=+] table [ y=amrc, x=rho, col sep=comma ] {data/simoNr4M8n3WuCorr08.txt}; 
\addplot    [color=blue!80!black, only marks, mark=o] table [ y=mrc, x=rho, col sep=comma ]  {data/simoNr4M8n3WuCorr08.txt};
\addplot    [color=blue!80!black, dotted] table [ y=wu_lb, x=rho, col sep=comma ]            {data/simoNr4M8n3WuCorr08.txt};
\addplot    [color=blue!80!black, dashed] table [ y=wu_ub, x=rho, col sep=comma ]            {data/simoNr4M8n3WuCorr08.txt};
\addplot    [color=blue!80!black, solid] table [ y=bound, x=rho, col sep=comma ]             {data/simoNr4M8n3WuCorr08.txt};
\end{semilogyaxis}
\end{tikzpicture}